\documentclass[review]{elsarticle}

\usepackage{lineno,hyperref}
\usepackage{amsmath,amssymb,mathrsfs}
\usepackage{graphicx,multirow,bm,rotating,float}
\usepackage{ntheorem}
\usepackage{geometry}
\modulolinenumbers[20]

\newcommand{\np}[2]{\ensuremath{\langle\,#1\mid #2 \,\rangle}}
\newcommand{\p}[1]{\ensuremath{\langle\, #1\,|}}
\newcommand{\cd}[1]{\ensuremath{|#1\,\rangle}}
\newcommand{\x}[1]{\ensuremath{|\,#1\rangle\,\langle\, #1 \,|}}

\newcommand{\dd}{\mathrm{d}}

\newtheorem{thm}{Proposition}
\theoremstyle{nonumberplain}
\theorembodyfont{\normalfont}
\newtheorem{proof}{Proof}
\theoremsymbol{\mbox{$\Box$}}

\journal{Journal of \LaTeX\ Templates}

\begin{document}

\begin{frontmatter}

\title{Quantum Probability Theoretic Asset Return Modeling:  A Novel Schrödinger-Like Trading Equation and Multimodal Distribution}


\author[mymainaddress,mysecondaryaddress]{Lin Li\corref{mycorrespondingauthor}}

\cortext[mycorrespondingauthor]{Corresponding author}
\ead{llin@ecust.edu.cn}

\address[mymainaddress]{Department of Finance, Business School, East China University of Science and Technology, Shanghai, 200237}
\address[mysecondaryaddress]{Risk-Center, ETH Z\"{u}rich, Switzerland, CH8092}


\begin{abstract}
Quantum theory offers a comprehensive framework for quantifying uncertainty, and many studies in quantum finance explore the stochastic nature of asset returns based on this theory, where the returns are likened to the motion of microscopic particles, adhering to physical laws characterized by quantum probabilities. However, such approaches inevitably presuppose that the changes in returns exhibit certain microscopic quantum effects, a presumption that may not guaranteed and has been criticized. In contrast to conventional approaches, this paper takes a novel perspective by asserting that quantum probability is primarily a mathematical scheme extending the classical probability from real to complex numbers, not exclusively tied to microscopic quantum phenomena. By directly linking the mathematical structure of quantum probability to traders' decisions and market trading behaviors, it circumvents the presupposition of quantum effects for returns and invocation the wave function. The phase in complex form of quantum probability serves as a additional element, capturing transitions between long and short decisions but also take information interaction among traders into account, This gives quantum probability an inherent advantage over classical probability in characterizing the multimodal distribution of asset returns. 
Through Fourier decomposition, we derive a Schrödinger-like trading equation, wherein each term corresponds explicitly to implications of market trading. The equation suggests discrete energy levels in financial trading, and returns follows the normal distribution at the lowest level. As the market shifts to higher trading levels, a phase transition occurs in the distribution of returns, leading to multimodality and fat tails. Empirical research on the Chinese stock market supports the existence of energy levels and multimodal distributions from this quantum probability asset returns model. 

\end{abstract}


\end{frontmatter}

\linenumbers

\section{Introduction}

In the modern financial markets, especially for the secondary markets, the observable asset price dynamics  are determined by unobservable interactions among many investors. Due to the heterogeneity of beliefs \cite{giglio2021five,Ji_2022}, diversity of  strategies \cite{farmer2002market, Hauser_2011} and incompleteness of information \cite{berrada2006incomplete, chen2021incomplete} that are during the interactions, the fluctuations in  price dynamics are inevitable, which give rise to the  randomness for the asset returns.  

The statistical laws of asset returns are characterized based on the stochastic process models in traditional way. Such models are represented by stochastic equations coupled with explicit terms to reflect  specific randomness that is governed by certain probabilistic laws. Typical stochastic process models to model asset returns broadly include two major categories, a class represented by stochastic differential equations in continuous-time setting and the other class dominated by time-series model (i.e., stochastic difference equations) in discrete-time setting.  Some popular models among the former includes the GBM model that suggest the returns obeying a lognormal distribution \cite{osborne1959brownian, black1973pricing}, the CEV model that considers the skewness of the returns distriubtion \cite{cox1996constant}, the Cox-Ingersoll-Ross model that emphasizes the mean-reverting nature of returns \cite{cox1985intertemporal}, the hyperbolic diffusion model of stock returns that captures the characteristics of leptokurtosis and fat-tails \cite{rydberg1999generalized}, the Heston model that introduces stochastic volatility coefficients \cite{heston1993closed}, the jump-diffusion model assuming the existence of Poisson jumps \cite{kou2002jump}, a Levy-COGARCH model extending generalized autoregressive conditional heteroskedasticity to continuous time\cite{kluppelberg2004continuous}, a stochastic cusp catastrophe model considering the bimodal shape of returns density \cite{barunik2009can}, and a fractal Brownian motion model to encode the long-memory dependence of returns \cite{rogers1997arbitrage}. The most popular models for discrete markets contain a family of GARCH models \cite{bollerslev1986generalized, hentschel1995all}, which can be regarded as various extensions  of the time-stochastic volatility Heston model in discrete time, and a family of multi-factor models, which attribute statistical patterns of returns to the time-varying behavior of multiple risk factors added with idiosyncratic noise, such as the most traditional CAPM models \cite{sharpe1964capital},  Fama-French three-factor model \cite{fama1993common}, Fama-French five-factor model \cite{fama2015five}, the Stambaugh-Yuan four-factor model \cite{stambaugh2017mispricing}, and the lately proposed DHS three-factor model that incorporates behavioral finance factors \cite{daniel2020short}.  In general, the stochastic process-based models emphasize that the randomness of returns mainly originates from the irreducible trading noise in the market.

In addition to the stochastic processes  modeling  based on classical probability theory, the quantum theory, armed with quantum probability formulation, offers another way to deal with randomness and unravel statistical laws. In terms of the historical context in which quantum theory was proposed, it seems to imply that quantum theory is a specific theory for the stochastic behavior of microscopic particles and is not applicable to characterize the statistical behavior of systems composed of macroscopic entities, such as financial markets. But leaving aside those microscopic properties of particles that have clear physical connotations \footnote{e.g., wave-particle dualism, uncertainty effects, quantum entanglement, quantum non-commutation, etc.}, quantum theory actually provides a complete mathematical scheme for solving the statistical laws of drift and fluctuation of observables by means of quantum probability. In this scheme, quantum probability is the mathematical concept which is the self-consistent extension of probability measure  in the form of real number to the measure in the form of complex number (e.q.  in quantum mechanics, quantum probability corresponds to the physical connotation of a complex wave function.) The quantum probability can be further decomposed into two components: modulus and phase, where the modulus (after squaring) directly corresponds to the probability size, while the introduction of the phase can help to model the potential periodicity for the change of probability and can conveniently capture the local multimodal nature of the distribution. In financial markets, continuous trading behavior with price fluctuations causes assets to form a distribution of holding costs at different price levels, often with multi-modal characteristics. The multi-modal holding cost implies the potential multi-modal distribution of the return \footnote{It should be noted that the empirical results mostly indicate that the asset return presents unimodal heavy-tailed distribution. In fact, the return distribution is often a mixture of a unimodal distribution and some multi-modal distribution in different market states. Hence,  the multi-modal distribution is easily ignored in the aggregate analysis due to the less time for the presence of  multi-modal distribution.} . From such a perspective, it is indeed reasonable and feasible to exploit quantum probability to model the behavior of asset returns.

In the past two decades or so, several asset pricing models based on the approach of quantum theory have been proposed in the literature, collectively known as models of quantum finance. In terms of modeling approach, scholars are mainly divided into two representative schools of thought: one school of thought is based on Feynman's path integral idea, which analogizes assets in financial market to particles in a quantum field and assumes that their evolution equations obey some kind of velocity equations for quantum field. The correlation of price change between different time points can be naturally inscribed by wave propagators at different points in the field. The financial model of path integral was first proposed by Baaquie \cite{baaquie2007quantum} and refined in subsequent studies for the quantum field theory model \cite{baaquie2013financial,baaquie2018quantum}. Another notable work is  path-independent quantum finance model for sensitivity analysis of path integral  \cite{kim2011sensitivity}. Because quantum probabilities in complex form have an additional phase property than classical probabilities, they bear more inherent advantages over traditional stochastic process models in describing the volatility and correlation of financial assets. However, because path integral necessarily requires the path endpoints have explicit payoffs, such methods can only model financial assets with contingent claims such as interest rate term structures, forward rates and their derivatives, and are less suitable for modeling asset returns like stock. In addition, the great computational burden requires by path integrals restricts the use of such methods in actual financial practice \cite{lee2020quantum}. Another school of thought is to start from the most fundamental Schrödinger equation in quantum theory and try to express the statistical laws of the evolution of financial assets as solutions of the wave function of the Schrödinger's equation. In the literature, motions and dynamics of price are analogous to quantum particle and randomness of price return is directly analogous to uncertainty displacement of a quantum particle from its equilibrium state, thus describing the statistical law by the anharmonic oscillation function.
 This school of research is more active, and most of them focus on the study of stock markets and the derivatives. For example, using the quantum wave functions for stock market analysis \cite{ataullah2009wave}
 ; assuming that the asset-volume-price relationship satisfies the Hamilton-Jacobi energy equation to study quantum probability \cite{shi2006does};  modifying the potential energy term of Schr\"{o}dinger's equation by distinguishing different traders to solve quantum probability of stock prices with anharmonic oscillator \cite{zhang2010quantum,gao2017quantum};  exploring the density of the return distribution by analogizing stock price limit up-limit down to quantum potential wells in  Schrödinger's equation \cite{meng2015quantum}; exploiting the theory of Bohmian potential approach, and pilot waves, with Schr\"{o}dinger's equation to incorporate behavioral financial factors into asset pricing \cite{choustova2007quantum,nasiri2018impact} and also an nonclassical oscillator model to characterize persistent fluctuations in stock market \cite{ye2008non}.

Although the  various  models in quantum finance are brilliant and with much highlight, however, most of the studies are modeled and solved by directly applying the relevant mathematical equations of quantum theory after simply making an analogy between the motion of microscopic particles and the motion of financial asset returns, without profoundly explaining the financial implications of these equations, rendering the meaning of many variables involved in the equations and solutions not well understood. For example, if the Schrödinger equation is used to study asset returns, what exactly does the Planck constant, which represents the smallest unit of energy in the equation, represent? What should the Hamiltonian and Lagrangian action quantities in it imply in financial terms? What is the theoretical implication of the practical correspondence of the phase term inside quantum probability? In short, the theoretical basis for applying the quantum probabilities expressed by the Schr\"{o}dinger equation to the study of financial variables has not been fully discussed. This has a lot to do with the rationale and necessity of applying quantum probability to financial modeling. This paper does not mean to contrast  applicable scenarios of  the quantum probability with classical probability , rather this paper will emphasize more on the relationship with the two. In fact, classical probability can be considered as a special case of quantum probability . By assigning some necessary interpretations, quantum probability will show particular advantages in capturing the statistical law of asset returns, as we will see shortly.


The road map of this paper is as follows. Section 2 reviews the basic concepts of quantum probability, emphasizing the connections and distinctions between quantum probability and classical probability. Section 3 explores how to assign financial interpretations to the mathematical structure of quantum probability based on traders' decisions and financial market trading behavior, especially to clarify the rationale for leveraging quantum probability through a discussion of the characteristics of active trading intentions. Section 4 first introduces the concept of $\omega$-market and employs this conceptual tool along with Fourier transforms to derive the quantum probability differential equation satisfied by return, termed the Schr\"{o}dinger-like trading equation. It is worth noting that this equation, although formally resembling the Schr\"{o}dinger equation in quantum theory, is derived entirely from \emph{non-microscopic} trading behavior. The involved variables, including kinetic energy or potential energy terms and energy levels, are echoed by explicit financial interpretations. Section discusses the model's solution and conducts empirical analysis using data from the Chinese stock market. The last section concludes.

\section{The basic conception of quantum probability}
Let us first discuss the connection between classical probability (CP) and quantum probability (QP). Recall that the CP theory is based on the notion of \emph{probability spaces}, which can be represented by a triple $(\Omega, \mathscr{F}, \mathbb{P})$.  $\Omega$ is the sampling space formed by a set of all possible single outcomes (elementary events). $\mathscr{F}$ is a $\sigma-$algebra on $\Omega$, that is a set including a family of subsets of $\Omega$. A certain subset of  $\Omega$ as an element in $\mathscr{F}$ is called a \emph{event} and only a event can be assign a number to indicate the size of the possibility of its occurrence, i.e. probability. $ \mathbb{P}$ is just the probability measure, a function defined on $\mathscr{F}$ to assign the probability on every events. 

The QP is also formed by a triple $(H, \mathscr{E}, \rho)$ \cite{focardi2020quantum,Yukalov_2016}, where $H$ stands for a Hilbert space  for which the inner product has been specified priori, and  the system with uncertainty is represented by a vector $h$ in $H$. $\mathscr{E}$ is the set of all permissible projection operators ${E}$  within $H$ for \emph{the vector} $h$. For a given orthogonal basis of $H$, orthogonal projections of $h$ represent \emph{events}. In particular, the projection to the orthogonal basis vector corresponds to  those of elementary events.  $\rho$ is called density operator. For a given $h$, $\rho$ is determined by $h$, and the quantity $\mathrm{tr}(\rho E)$ gives the probability of event ${E}$ in system $h$.

QP may seem to have a completely different mathematical structure from CP, and it is more abstract and obscure, not as intuitive as the definition of CP. Nonetheless, the probabilities defined in these two ways with different triples have a great similarity.  First,  just as $\Omega$ defines all the elementary events, $H$, together with a given orthogonal basis, also identifies all the elementary events. Second, for a given $h$ in $H$, the $\mathscr{E}$ of QP is mathematically isomorphic to the $\mathscr{F}$ of CP in terms of the identifying  events. The mutually exclusive events in $\mathscr{F}$, i.e., $A_1, A_2, A_1\cap A_2=\phi$  correspond to the two mutually orthogonal projection operators in $\mathscr{E}$, i.e., $E_1, E_2, E_1\bot E_2$. Third, for a given event, the size of probability on it is determined by the nature of system rather than the structure of $\mathscr{F}$ or $\mathscr{E}$. In CP, different system refers to different probability measure $\mathbb{P}$ that determines the probability. While in QP, different system corresponds to different the density operator that determines the probability.   The last but most important, the probability measure $\mathbb{P}(E)$ in CP can be reformulated in the form of $\mathrm{tr}(\rho E)$. The above similarities can be seen more clearly  with the Dirac's bra-ket notation system, as follows.

without loss of generality, consider the Hilbert space $H$ is separable and complex. The given orthogonal basis is $\{\cd{e_i}\}_{i=1}^{\infty}$, where $\forall i\neq j, $, the inner product of $\cd{e_i}$ and $\cd{e_j}$ is zero, i.e. $\np{e_i}{e_j}=0$ and $\mathrm{Span}(\cd{e_1},\cd{e_2}\cdots,\cd{e_{\infty}})=H$.  Suppose the uncertainty system $h \in H$ that we concern is  the vector $\cd{\psi}$ in $H$.  An orthogonal projection operator which project $\cd{\psi}$ into the subspace of $\text{Span}(\cd{e_i})$ can be represented by $\x{e_i}$, since, 
\begin{equation*}
	\biggl(\x{e_i}\biggr)~\cd{\psi}=\np{e_i}{\psi}\cdot \cd{e_i}~.
\end{equation*}

There exists a mutual translation from the language of CP to the language between QP. Obviously, all  operators $\x{e_i}$ constitute the set of elementary events. Hence, $\x{e_i} \in \mathscr{E}$ can be translated into $\{\omega_i\}\in \mathscr{F}$ in the language of CP, where the sampling space is $\Omega=\{\omega_i\}_{i=1}^{\infty}$.  For arbitrary non-elementary event $A\in  \mathscr{F}$, e.g. $A=\{\omega_i \mid i=1,2,3\}$, it can be translated into a orthogonal projection operator $E \in \mathscr{E}$, which is written by
$
	E=\sum_{i=1}^3 \x{e_i} ~.
$

If two events in different probability can be translated to each other, we denote as $A \overset{t}{\sim} E$. Now we show that any two mutually exclusive events $A_1, A_2 \in \mathscr{F}$  are equivalent to two orthogonal events $E_1, E_2 \in \mathscr{E}$, if $A_k \overset{t}{\sim} E_k, k=1,2$. Suppose $A_1=\{\omega_i \mid i\in \mathbb{S}_1\}$ and $A_2=\{\omega_j \mid j\in \mathbb{S}_2\}$, where $\mathbb{S}$ denote a set of numbers. Since $A_1, A_2$ are mutual exclusive, $A_1 \cap A_2=\phi$. And they are respectively translated to 
\[
E_1=\sum_{i\in \mathbb{S}_1} \x{e_i},\quad E_2=\sum_{j\in \mathbb{S}_2} \x{e_j}, ~~\qquad  \mathbb{S}_1 \cap  \mathbb{S}_2=\phi ~.
\]
Then the inner product of the two projections $\cd{p_1}=E_1\cd{\psi}$ and $\cd{p_2}=E_2\cd{\psi}$ equals to 
\begin{align*}
	\np{p_1}{p_2} &= \left(\sum_{i\in \mathbb{S}_1} \np{e_i}{\psi}^*\p{e_i}\right) \sum_{j\in \mathbb{S}_2} \cd{e_j}\np{e_j}{\psi}\\
	&=\sum_{i\in \mathbb{S}_1} \sum_{j\in \mathbb{S}_2}  \np{e_i}{\psi}^*\np{e_j}{\psi} \np{e_i}{e_j}=0~.
\end{align*}
So $E_1$ and $E_2$ are mutually orthognal. The last equation holds because  $\mathbb{S}_1 \cap  \mathbb{S}_2=\phi $. 

Next, suppose the system is assigned the probability measure $\mathbb{P}$ in CP. For a even $A=\{\omega_i\mid i \in \mathbb{S}\}\in \mathscr{F}$, its probability is quantized as
\begin{equation}\label{qwer}
	\mathbb{P}(A)=\sum_{ i \in \mathbb{S}} \mathbb{P}(\omega_i)~. 
\end{equation}
While, in QP, the system $\cd{\psi}$ and  the density operator $\rho_{\psi}=\x{\psi}$ are bounded together. Thus, for every elementary event $e_i = \x{e_i}$, the probability is quantized as 
\begin{align}
\mathrm{P}(e_i)=\mathrm{tr}(\rho_{\psi} e_i)&=\sum_{j=1}^{\infty} \p{e_j}\biggl(\x{\psi}\biggr)\biggl(\x{e_i}\biggr)\cd{e_j}\notag\\
	&=\sum_{j=1}^{\infty} \np{e_j}{\psi} \np{\psi}{e_i}\np{e_i}{e_j}= \np{e_i}{\psi}\np{\psi}{e_i}\label{wert}~.
\end{align}
The last equation holds because $\forall i\neq j, \np{e_i}{e_j}=0$. As $A \overset{t}{\sim} E$, 
\begin{align}
	\mathrm{P}(E)=\mathrm{tr}(\rho_{\psi}E)&=\sum_{j=1}^{\infty} \p{e_j}\biggl(\x{\psi}\biggr)\biggl(\sum_{i \in \mathbb{S}} \x{e_i}\biggr)\cd{e_j}\notag\\
	&=\sum_{j=1}^{\infty}\sum_{i \in \mathbb{S}} \np{e_j}{\psi} \np{\psi}{e_i}\np{e_i}{e_j}=\sum_{i \in \mathbb{S}} \np{e_i}{\psi}\np{\psi}{e_i}\label{erty}
\end{align}

Compare the forms of eq. (\ref{qwer}) and (\ref{erty}) and note eq.(\ref{wert}), one can see that  $\mathbb{P}(A)$ can be reformulated $\mathrm{tr}(\rho_{\psi}E)$, given $A \overset{t}{\sim} E$. It is worth noting that the real-valued function $\mathbb{P}: \Omega\mapsto [0,1]$ in CP will be further decomposed into two conjugate complex functions $\psi(\cdot)=\np{\cdot}{\psi}$ and $\psi^*(\cdot)=\np{\psi}{\cdot}$, i.e., probability amplitude function . According to eq.(\ref{wert}), the probability is $P(\cdot)=\phi(\cdot)\phi^*(\cdot)$. In the case of $\phi$ being real-valued, $\psi=\psi^*$ and $P(\cdot)=\psi^2(\cdot)$, QP is subtly reduced to the CP. However, when $\psi$ is a true complex function, $\psi$ can provide an additional information because $\psi$ contain the measurement of phase other than modulus. From this point of view, QP is a sort of extension of CP. It \emph{should not be} restricted to applications to microscopic quantum phenomena only. It is this extension that gives advantage of using QP to capture the statistical patterns of asset returns. Another noteworthy point is that the function $\psi$ is precisely the coefficient of $\cd{\psi}$ on each orthogonal basis vector when projected within the Hilbert space $H$, i.e., 
\begin{equation}\label{qwert}
	\cd{\psi}=\left(\sum_{i=1}^{\infty}\x{e_i}\right)\cd{\psi}=\sum_{i=1}^{\infty}\np{e_i}{\psi}\cdot \cd{e_i}~.
\end{equation}

%
%
%

\section{The quantum probability framework for statistical law of asset returns}
When modeling using the classical probabilistic framework, asset returns are traditionally treated as a random $\widetilde{r}$ variable, and the statistical law of returns are characterzied by a distribution density function (or mass function) $f_{\widetilde{r}}(r)$ that reveals the relative frequency (Frequencist's view ) or likelihood (Bayesian's view) of taking a particular value $r$ for the return. 


However, $f_{\widetilde{r}}(r)$ does not record the market's active trading intentions (ATI) at corresponding return levels. For example, $f_{\widetilde{r}}(r=2\%)=0.1$ only indicates that at a specific moment, the probability of the return reaching the $2\%$ is $0.1$. If we denote the closing price as $P^*=P_0(1+2\%)$, where $P_0$ is the initial price, then $f_{\widetilde{r}}(r=2\%)=0.1$ also signifies the probability of the closing price increasing from $P_0$ to $P^*$ is $0.1$. It's important to note that $P^*$ is determined by the specific active buy or sell orders matched at that moment. However, $f_{\widetilde{r}}(r)$ remains silent on the property of the ATI shaping the price $P^*$.

Imagine that at the $P^*$ level, the trading volume is 100. In this scenario, three possibilities exist:
\begin{itemize}
	\item Market's active trading intention is ``short": There are more than 100 buy orders at $P^*$, and a trader actively sells 100 shares.
	\item Market's active trading intention is ``long": There are more than 100 sell orders at $P^*$, and a trader actively buys 100 shares.
	\item Market's active trading intention includes both ``long" and ``short": Traders buy 100 units at the market price of $P^*$, while another group of traders simultaneously sells 100 units at the market price, achieving exactly 100 shares traded.
\end{itemize}

People cannot distinguish which scenario has occurred in terms of $f_{\widetilde{r}}(r=2\%)=0.1$. However, the property and intensity of ATI are crucial for modeling returns. In reality, the different property and intensity of ATI impact the probability of achieving specific returns and the dynamic changes in prices. When liquidity is sufficient, having a sufficient number of passive orders to support the price level, higher-intensity of ATI  (regardless of long or short) will increase the number of trades at that level, thus raising the size of probability of the corresponding return. When liquidity is insufficient, with an inadequate number of passive orders , the property of active trading will determine whether prices are pushed higher, increasing returns, or prices are pulled lower, decreasing returns. As discussed in subsection 2.1, the function $\psi(\cdot)$ in quantum probability framework are allowed to carry addition information. In fact, the concealed property and intensity of ATI can be precisely encoded uniformly within the quantum probability framework, as we will see shortly.

 In the quantum probability framework, we employ $\cd{\psi}$ to denote the market state and the asset return at time-point $t $ taken value of $r$ is considered as an elementary event $\x{r,t}$ . According to eq. (\ref{qwert}), the market can be decomposed as
 \begin{equation*}
 	\cd{\psi}=\left(\sum_{r,t}^{\infty}\x{r,t}\right)\cd{\psi}=\sum_{i=1}^{\infty}\np{r,t}{\psi}\cdot \cd{r,t}~.
 \end{equation*}
 The complex-valued probability amplitude function,  $\Psi(r, t)=\np{r,t}{\psi}$,  can be further specified as $\phi(r,t)\, e^{i\, \theta(r,t)}$. The distribution density function $f(r,t)=\Psi(r,t)\Psi^*(r,t)=|\Psi(r,t)|^2=\phi(r,t)^2$.  This paper only considers  the case of market where the system  are \emph{i.i.d} with respect to time. Hence, $\Psi(r,t)$ can be simplified as $\Psi(r)=\phi(r)\, e^{i\, \theta(r)}$, independent of $t$. And $f(r,t)$ reduces to $f(r)=|\Psi(r)|^2=\Psi(r)\Psi^*(r)$.  $\Psi(r)$ can be represented  by a spiral curve on the three-dimensional columnar coordinate space $r-\phi(r)-\theta(r)$, where the vector pointing from the origin to the curve is like a clock hand that scales.  As
$r$ changes, the length of the clock hand will undergo a corresponding alteration, accompanied by a certain degree of rotation. The former and the latter are characterized by $\phi(r)$  and 
$\theta(r)$, respectively. Figure 1 gives an illustration.

 \begin{figure}[t]
 	\centering
 	\includegraphics[width=12cm]{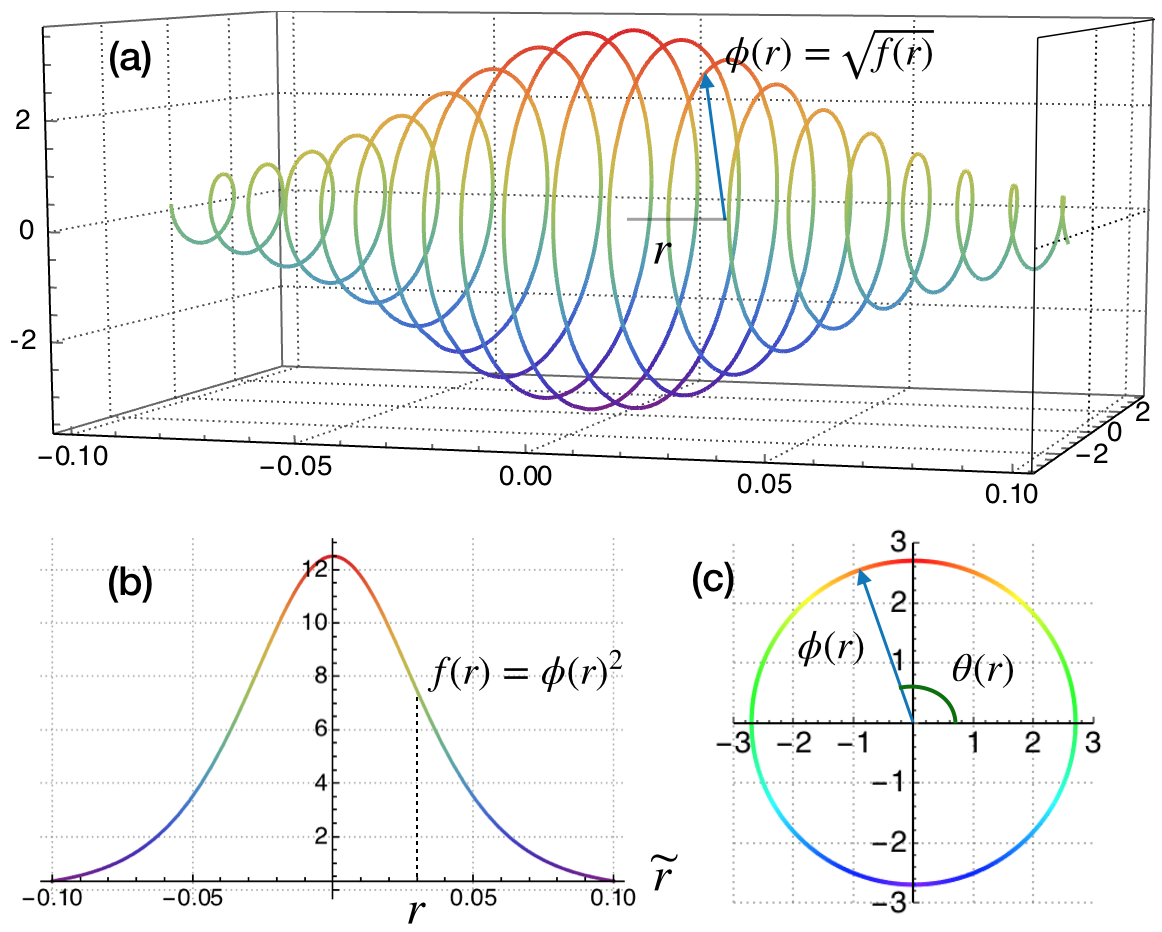}
 	\caption{Subfigure (a) gives the $\Psi(r)$ curve in the three-dimensional space characterizing the  information of return. Subfigure (b) shows  the density function of return that comes from $\Psi(r)$ based on side view plot of subfigure (a), where the density function is equal to the square of modulus of  $\Psi(r)$. Subfigure (c) shows the another side view plot of subfigure (a), which is a hypothetical complex plane, for possible values $r$, with the vector corresponding to the complex number $\Psi$ with length $\phi(r)$, while being rotated by a certain angle $\theta(r)$ .}
 \end{figure}
 
   
 We propose to use $\phi(r)$ to represent the ATI intensity at a specific $r$. A larger ATI intensity implies a higher probability of trading at that $r$, contributing to a higher probability density on levels of $r$. The function $f_{\widetilde{r}(r)}=\phi(r)^2$ satisfies this monotonically increasing relationship.  Meanwhile, we propose to employ $\theta(r)$ to characterize the property of ATI. It comes from two motivations.

\begin{figure}[t]
\centering
	\includegraphics[width=11.5cm]{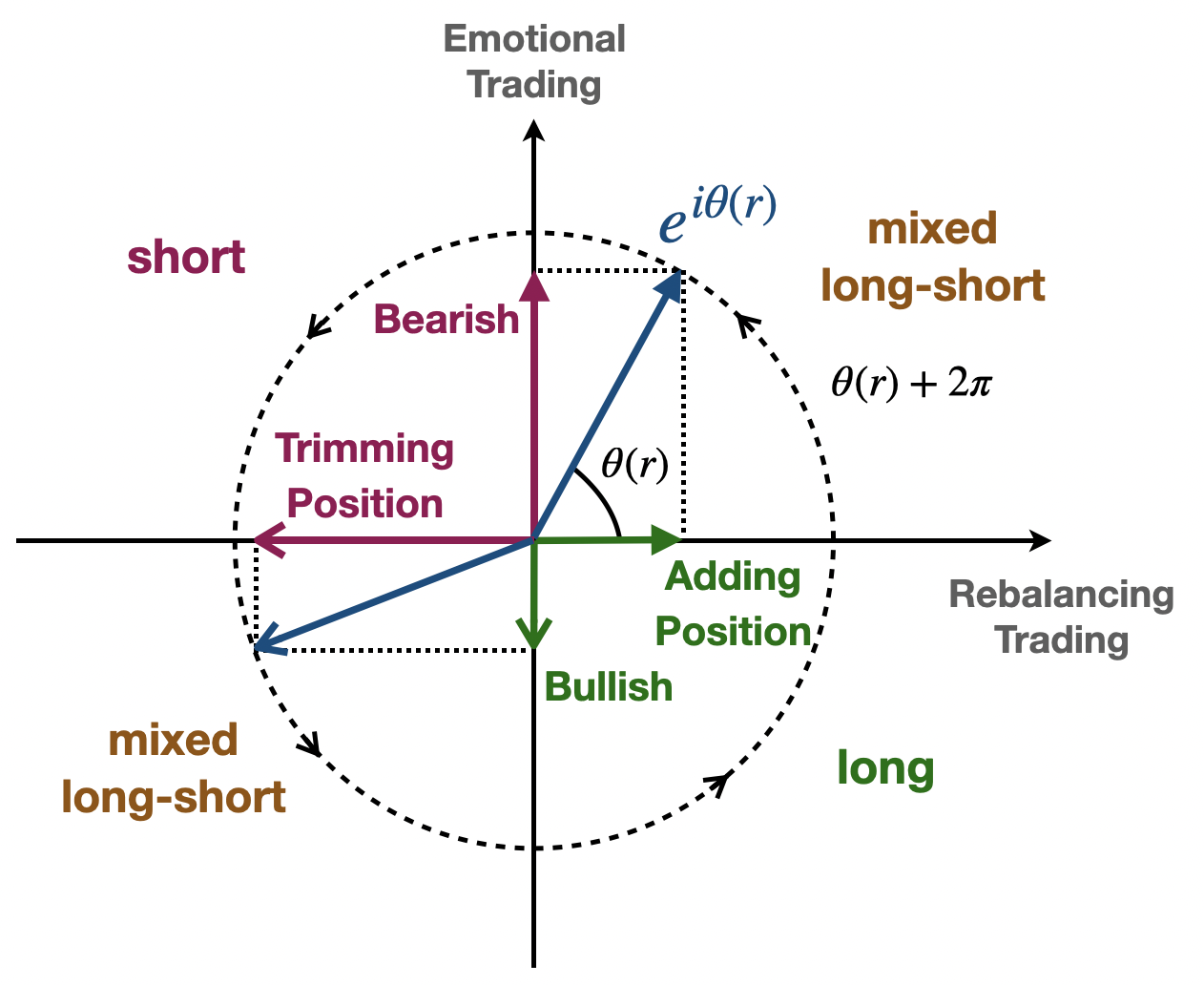}
	\caption{The ATI Plane. In the ATI plane, complex numbers with a magnitude of 1 are employed to characterize the properties of ATI. The figure illustrates the properties of ATI corresponding to two different market states, denoted by distinct arrow styles.The real and imaginary parts of these complex numbers signify rational position rebalancing intention and emotional trading intention, respectively. Positive real part denote position adding demand, while negative real part indicate position trimming demand. On the imaginary axis, positive values represent bearish sentiment, and negative values represent bullish sentiment. In the second quadrants, the ATI is short, while in the fourth quadrants, it is long. In the first and third quadrants, the ATI property become mixed long-short. As the complex number completes a full rotation in the ATI, it undergoes a cycle oscillating from long (short) to mixed long-short and back to long (short). During counterclockwise rotation, market with a long(short) ATI will transition first into a pure bearish(bullish) state and then into a bullish(bearish) state; conversely, during clockwise rotation, the sequence is precisely reversed.}
\end{figure}

\emph{The first motivation}  arises from the fact that the investors in the market often oscillates between the  intention of  ``long" and ``short"  cyclically with unidirectional change in $r$. Let's illustrate this behavior with a representative investor $A$, who faces increase of $r$. The analysis  where the $r$ unidirectionally decreases is similar. Suppose at $P=P_0, r=r_0$, $A$'s ATI is ``long". In a typical cycle, $A$ will undergo the following four stages, as Figure 2 illustrates.
\begin{itemize}
\item[(1)] As the price $P$ increases, return $r$ also increases. But as $r$ increases, the position rebalancing intention of $A$, manifesting as ``position adding" demand for the stock  gradually decreases, while the ``bearish" sentiment  correspondingly increases. When the price $P$ increases to a certain value, for example, $P_1 > P_0$ (as $r\to r_1$), the ``bearish" sentiment will reach its maximum, the ``position adding"  shrinks to zero simultaneously. Beyond the point,  $A$'s ATI turns to be  ``short".
\item[(2)]  With further increase in $P$, $A$'s position rebalancing intention turns negative and increases negatively. It indicates an increase in $A$'s ``position trimming" intention. $A$'s ATI remains ``short". However, as the price continues to increase, $A$'s sentiment leads him to start doubting the price forecast, and the "bearish" sentiment gradually decreases. When $P$ continues to increase to a higher level, for example, $P_2 > P_1$ (as $r\to r_2$), the ``bearish" sentiment  eventually vanishes.
\item[(3)]  When $P$ increases more, $A$'s ``bullish" sentiment turns up. However, the ambiguity in price uptrend leads $A$ to rationally maintain the "position trimming" intention. As $P$ increases, the ``position trimming" intention becomes smaller and will be eventually reduced to zero when the price increases to a higher level, for example, $P_3 > P_2$ (as $r\to r_3$). $A$'s ATI turns back to be ``long".
\item[(4)]  When $P$ increases further, $A$ eventually forms a rational rebalancing intention based on momentum trading strategy. The "position adding" intention turns up and increases, and $A$'s ATI remains "long."  However, with the ongoing increase in price, $A$ becomes increasingly cautious, leading to a gradual contraction of the "bullish" sentiment.  When the price reaches a certain level, for instance, $P'_0 > P_3 > P_0$ (as $r\to r'_0$), $A$'s "bullish" sentiment completely diminishes. Beyond this point, the "bearish" sentiment will start to reemerge, and $A$ will enter a new cycle.
\end{itemize}


Using the complex number $e^{i\theta(r)}$ to quantify ATI allows for a straightforward characterization of the narrative above. In this context, the four stages conveniently correspond to the four quadrants on the complex plane where $e^{i\theta(r)}$ resides. The variation in $\theta(r)$ induces changes in the direction of the complex number, signifying distinct ATI properties. A complete cycle corresponds to a rotation of the vector in the complex plane by $2\pi$, where $\theta(r)$ increases by $2\pi$. The real axis represents the (rational) intention of position rebalancing , while the orthogonal imaginary axis represents intention of emotional trading. The positive direction on the real axis indicates "position adding" , and the positive direction on the imaginary axis indicates "bearish" sentiment.  Both "adding/trimming" demand and "bearish/bullish" sentiment are simultaneously represented in ATI, expressed as real and imaginary parts of  $e^{i\theta(r)}$ . This can be further understood as the ATI is a superposition of intentions for both rational rebalancing and emotional trading with some weights. 
Defining two special market states, where investors collectively exhibit unit-intensity position adding demand denoted by $\cd{\text{adding}}$, and unit-intensity bearish sentiment denoted by $\cd{\text{bearish}}$. They satisfy $\np{r}{\text{adding}}=1$ and $\np{r}{\text{bearish}}=i$. According to the superposition effect, the market state satisfies,
\[
\cd{\psi}=a(r)\cdot \cd{\text{adding}}+b(r)\cdot \cd{\text{bearish}}~.
\]
Since,
\[
\phi(r)e^{i\theta(r)}=\np{r}{\psi}=\p{r}\biggl[a(r) \cdot \cd{\text{adding}}+b(r)\cdot\cd{\text{bearish}}\biggr]=a(r)+i~ b(r)
\]
evidently,
\begin{equation}
	\begin{cases}
	a(r)=\phi(r)\cos \theta(r)\\
	b(r)=\phi(r)\sin\theta(r)
\end{cases}
\end{equation}

 We refer to the complex plane where $e^{i\theta(r)}$ resides as the ATI plane (Figure 1(c)). Figure 2 elaborates on the implications of each element in the ATI plane in detail.

\emph{The second motivation} arises from the capability of operations in complex number to set forth information interactions among active traders and their impact on ATIs. Let's exemplify this with a rather simple scenario. For pedagogical purposes, assume the market comprises only two types of active traders, denoted as Alice and Bob. Passive traders are supposed to exist in the market to provide liquidity. 


When the market only consists of  Alice,  the markets are referred to as Market 1; Correspondingly, the market only consists of Bob are referred to as Market 2. Obviously, Market 1's ATI is equivalent to Alice's ATI, and similarly, Market 2's ATI is equivalent to Bob's ATI. Figure 3 illustrates the ATIs of the two markets at a specific price level (return at $r$).  As it shows, Market 1's ATI is a superposition of "position adding" and "bearish" sentiment, while Market 2 only has a "position trimming" demand. Additionally, the components of the ATIs on the real axis in Market 1 and Market 2 oppositely. 
As discussed earlier, the probability density of return at $r$ in the market is determined by the number of trades, which is influenced by ATI. Therefore, without considering the information interaction between Alice and Bob, the number of trades induced by their presence in market simultaneously at $r$ should be determined by the sum of the intensities of their ATIs. At this point, the property of the new market's ATI (ATI direction) has no impact on the probability density.
\begin{figure}
	\centering
	\includegraphics[width=15cm]{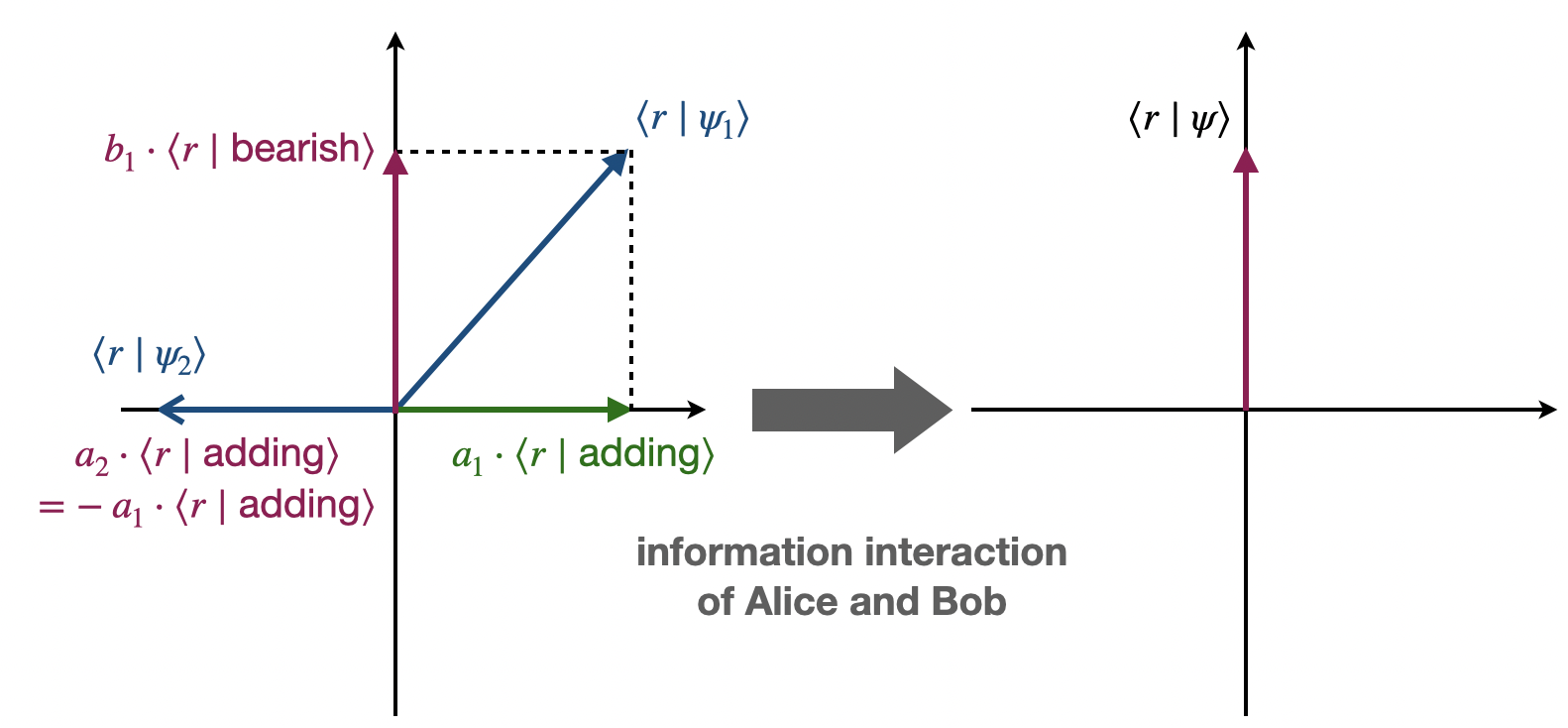}
	\caption{A simple illustration on the impact of information interaction on market ATI.The left panel of the ATI plane illustrates the ATI of the market when there is only one active trading investor. For Market 1 with only Alice, the ATI at $r$ is given by $\np{r}{\psi_1} = a_1 \cdot \cd{\text{adding}} + b_1 \cdot \cd{\text{bearish}}$. For Market 1 with only Bob, the ATI at $r$ is $\np{r}{\psi_1} = a_2 \cdot \cd{\text{adding}}$, where $a_2 = -a_1$, indicating that in terms of position rebalancing intention, Alice and Bob's ideas are precisely opposite with equal intensity.The consequence of interaction is depicted in the right panel, showing the market's ATI when both Alice and Bob are present. The interaction results in the simultaneous withdrawal of position rebalancing intentions by Alice and Bob. Contrary intentions do not lead to trades. Consequently, the market's ATI is left with the bearish emotional trading intention, which contributes transactions.}
\end{figure}


However, when there is information interaction between the two, the situation changes. Information interaction induces alterations in their respective ATIs. Knowing that Bob intends to trim his position, Alice, concerned that Bob possesses private information she lacks, reduces her "position adding" demand. Conversely, Bob, aware of Alice's "position adding" demand, reduces his "position trimming" demand, suspecting a potential information deficit.
Under information interaction, conflicts in their position rebalancing intentions neutralize, and the rational trading that each would have independently generated disappears. This effect aligns with the well-known notion of the Financial No-Trade Theorem (\cite{notrading}). The final outcome is that the new market's ATI, under information interaction retains only Alice's bearish emotional trading intention \footnote{It is important to note that, due to the independence of the components related to rational trading and emotional trading within ATI, there is no information interaction. Therefore, the position adding demand on the real axis and the bearish sentiment on the imaginary axis do not offset each other. For instance, Bob's rational position trimming demand does not intensify due to the bearish sentiment that Alice possesses.}. The Interaction not only alters the properties of the market ATI but also its intensity, thereby reshaping the distribution behavior of returns. Handling such effects within the real number is challenging, but they can be elegantly characterized using complex addition and multiplication based on the ATI plane. Assuming the ATI of Market 1 and Market 2 is $\Psi_1 = \phi_1(r) e^{i\theta_1(r)}$ and $\Psi_2 = \phi_2(r) e^{i\theta_2(r)}$ respectively. The interaction leads the new market's ATI to be $\Psi = \Psi_1 + \Psi_2$,  and thus, 
\begin{align}\label{yuiop}
	f(r)&=\Psi(r)\Psi^*(r)=[~\Psi_1(r)+\Psi_2(r)~][~\Psi^*_1(r)+\Psi^*_2(r)~]
	\notag\\
	&=|\phi_1(r)|^2+|\phi_2(r)|^2+\phi_2(r)\phi_1(r) e^{i~[\theta_2(r)-\theta_1(r)]}+\phi_1(r)\phi_2(r) e^{i~[\theta_1(r)-\theta_2(r)]}\notag\\
	&=|\phi_1(r)|^2+|\phi_2(r)|^2+2\phi_1(r)\phi_2(r)\cos[\,\theta_1(r)-\theta_2(r)\,]~.
\end{align} 


The third term given by \eqref{yuiop} represents the effect of information interaction. Recall that for Market 1 and 2, $f_1(r)\sim |\phi_1(r)|^2$ and $f_2(r)\sim |\phi_2(r)|^2$ (ignoring normalization constants). In the absence of information interaction, the third term is zero. Therefore, the $f(r)$ of the new market is the sum of the contributions from Alice and Bob (independently in their respective markets), resulting in a mixed distribution.When $\theta_1(r)-\theta_2(r)\neq 0$, interaction effects between Alice and Bob emerge. When $\theta_1(r)-\theta_2(r)\in (2k\pi-\dfrac{\pi}{2},~2k\pi+\dfrac{\pi}{2}), k\in \mathbb{Z}$, the term is positive, indicating an increase in probability density. Intuitively, if the two traders have similar ATI properties at a given $r$, their intentions will reinforce each other, prompting synchronized trading activities at $r$, thus increasing the chances of successfully achieving $r$. When $\theta_1(r)-\theta_2(r)\in (2k\pi+\dfrac{\pi}{2},~2k\pi+\dfrac{3\pi}{2}), k\in \mathbb{Z}$, the term is negative, decreasing the probability density given $r$, which characterize the Financial No-Trade effect  where intentions oppose each other. Specifically, the term reaches its maximum when $\theta_1(r)-\theta_2(r)=0$. By contrast, when $\theta_1(r)-\theta_2(r)=k\pi$, the term reaches its minimum.


The analysis above can be generalized to the case of multiple traders, revealing that, taking into account the effect of information interaction, the market's ATI is essentially an aggregation of each trader's ATI. This can be expressed as:

\begin{equation}\label{E:totATI}
	\Psi(r)=\phi(r)e^{i\theta(r)}=\sum^N_{i=1}\phi_i(r)e^{i\theta_i(r)}= \sum_{i=1}^{N}\Psi_i~,
\end{equation}
where $N$ represents the total number of  investors of active trading in the market, and $\Psi_i$ is the ATI of investor $i$. 
We refer to Equation \eqref{E:totATI} as the \emph{principle of decomposing the market's total ATI } among traders.

\section{The model for asset returns}
\subsection{The $\omega$- market}

Given the ATI of the market, its Fourier decomposition yields
\begin{equation}\label{E:fuliye}
	c(\omega)=\dfrac{1}{{2 \pi}}\int_{-\infty}^{+\infty} \Psi(r)\,e^{-i\,\omega r}\dd\, r
\end{equation}
We will immediately discern the implications of the coefficient $c(\omega),$ necessitating the introduction of the $\omega$-market concept. Performing the Fourier inverse transform on the above equation, we obtain
\begin{equation}\label{E:ifuliye}
	\Psi(r)=\int_{-\infty}^{+\infty} c(\omega)\,e^{i\,\omega r}\dd\, \omega
\end{equation}
In accordance with Equation \eqref{E:totATI}, the above expression elucidates the principle of decomposing the market's total ATI among traders. All active traders in the market are categorized into different groups based on distinct $\omega$ values. Traders within each group are collectively referred to as $\omega$-traders. Correspondingly, the market formed by these traders is termed the $\omega$-market.

From the rotational behavior of ATIs in the ATI plane, it is evident that $\omega$-trader's ATI with the same $\omega$ value synchronously change direction with variations in $r$. Simultaneously, each $\omega$-trader in the $\omega$-market has a target return rate and emotional cycle. This is because, on one hand, $\omega$ determines the trader's internally setting target return rate. As the return rate $r$ increases by $\frac{2\pi}{\omega}$, the ATI completes one revolution, undergoing a rebalancing cycle based on the target return rate $\bar{r}=\frac{2\pi}{\omega}$. On the other hand, $\omega$ also dictates the emotional cycle of relative return rate variations for traders. This is evident as the ATI undergoes one complete rotation of bullish/bearish in the ATI plane, corresponding to a cycle of $\frac{2\pi}{\omega}$.

%
%
Consider the ATI of the $i$-th $\omega$-trader in the $\omega$-market, expressed as $\phi_i(r) e^{i \omega r+\theta_i}$. Here, $\theta_i$ represents the trader's heterogeneous trading intention at the reference point $r=0$, attributed to their prior holdings or accumulated sentiment.
Following the trader decomposition principle, the collective ATI of $\omega$-traders is given by $\phi_i(r) e^{i \sum_i\theta_i}e^{i \omega r}$. Comparing this with Equation \eqref{E:ifuliye}, it is evident that the $c(\omega)=\phi_i(r) e^{i \sum_i\theta_i}$ where the modulus of $c(\omega)$ represents the ATI intensity of the $\omega$-market, and its phase conveys the collective ATI property of the $\omega$-market at $r=0$.

It is crucial to note that Equation \eqref{E:ifuliye} does not assert that each \emph{actual} active trader in the market must correspond to an \emph{actual} individual $\omega$-trader. A $\omega$-trader should be understood as an \emph{agent} in the market with trading intent. This could be a trading strategy of atrader, a certain set of trading rules, a specific trading seat or account, or even represent an actual trader themselves.

\subsection{The Shr\"{o}dinger-like trading equation}

Firstly, it is noteworthy that regardless of the direction of ATI vector, it represents one trading decision. Thus, as the vector rotates by $\Delta\theta$ in the ATI plane, the area swept by the vector will be statistically proportional to the number of trades. For an $\omega$-market, the average number of trades contributed by a change in  $\Delta r$ is proportional to $\omega \Delta r$. Let $f_r$ denote the average shares of active trading per decision for a change $\dd r$ initiated at $r$. Therefore, the contribution of $\omega$-market to the trading volume for a change $\dd r$ is given by $\dd Q_{\omega}=f_r\, \omega\, \dd r$, which is always positive\footnote{When $\omega$ <0, $\dd r$ <0, hence $\dd Q_{\omega}$ is also positive.}. Market depth varies at different yield rates, so $f_r$ is a function of $r$.

On the other hand, it is crucial to recognize that when return is at $r$, the $(\omega+\dd \omega)$-market will make more trading decisions than the $\omega$-market, proportional to $r \dd \omega$. The additional trading volume contributed is also proportional to $ r \dd \omega$. When $r$ remains unchanged, let the average shares per additional active trade be denoted as $g$. Therefore, at  $r$, the difference in trading volume contributed by two $\omega$-markets with a difference $\dd \omega$ is $\dd Q_{r}=g\, r\,\dd \omega$. Since there are no differences in the order placement mechanisms across different $\omega$-markets, $g$ is a constant, independent of $\omega$ and $r$. Assuming there exists a specific $\bar{\omega}$ and $\bar{r}$ such that $\dd Q_{\omega}=\dd Q_r$, i.e., $f_r \,\overline{\omega}\, \dd r=g\, \overline{r}\,\dd \omega$, we have,
\begin{equation}
    f_r= h \dfrac{\dd \omega}{\dd r},\quad \text{where},\,\, h:=\dfrac{g\overline{r}}{\overline{\omega}}.
\end{equation}
 
According to the above expression, the trading volume contributed by the $\omega$-market during the transition from $0$ to $r$ in asset return is given by
 \begin{equation}
 	Q_{\omega}=\int_{0}^t\dd Q_{\omega}=\int_{0}^t f_r\, \omega\, \dd r = \int_{0}^{\omega} h \dfrac{\dd \omega}{\dd r}\cdot \omega\, \dd r= \dfrac{h}{2}\omega^2
 \end{equation}
Furthermore, the following theorem holds.

\begin{thm}

Assuming the complex-valued probability density of asset returns at $r$ is denoted as  $\Psi(r)$,  the expected realized trading volume in the process of attaining this return is expressed as:

\begin{equation}\label{E:thm1}
\overline{Q}=\int_{-\infty}^{+\infty}\Psi^*(r)\left(-\dfrac{h}{2}\dfrac{\dd^2}{\dd r^2}\right)\Psi(r)\dd r
\end{equation}

Here, $\psi^*(r)$ represents the complex conjugate of $\Psi(r)$.

\end{thm}

\begin{proof}
According to the principle of decomposing the market's total ATI, the probability for the contribution of a $\omega$-market to the total trading volume is given by $|c(\omega)|^2 = c^*(\omega)c(\omega)$. Therefore,

\begin{align*}
	\overline{Q}&=\int_{-\infty}^{+\infty} Q_{\omega} c^*(\omega)c(\omega)\dd \omega =\int_{-\infty}^{+\infty} c^*(\omega)\dfrac{h}{2}\omega^2  c(\omega)\dd \omega \\
& = \int_{-\infty}^{+\infty}\left(\int_{-\infty}^{+\infty} \Psi(r)\,e^{-i\,\omega r}\dd\, r\right)^* \dfrac{h }{2}\omega^2 c(\omega)\dd \omega\\
&= \int_{-\infty}^{+\infty} \Psi^*(r)\left(\int_{-\infty}^{+\infty} \dfrac{h }{2}\omega^2\,e^{i\,\omega r} c(\omega)\dd \omega\right)\dd \,r\\
&=\int_{-\infty}^{+\infty} \Psi^*(r)\left(\int_{-\infty}^{+\infty} \dfrac{h }{2}\left[\left(-i\dfrac{\dd }{\dd r}\right)\left(-i\dfrac{\dd }{\dd r}\right)\,e^{i\,\omega r}\right] c(\omega)\dd \omega\right)\dd \,r\\
&= \int_{-\infty}^{+\infty} \Psi^*(r)  \left(-\dfrac{h }{2}\right) \left(\dfrac{\dd^2}{\dd r^2}\right) \left[\int_{-\infty}^{+\infty} e^{i\,\omega r} c(\omega)\dd\omega\right]\dd\,r\\
&=\int_{-\infty}^{+\infty}\Psi^*(r)\left(-\dfrac{h}{2}\dfrac{\dd^2}{\dd r^2}\right)\Psi(r)\dd r	\qquad\qquad\qquad\qquad\qquad\qquad\qquad\qquad \Box
\end{align*}
 
\end{proof}

From the structural perspective of Equation (\ref{E:thm1}), due to the fact that $\Psi^*(r)\Psi(r)$ precisely represents the probability density of achieving a return rate of $r$, the differential operator $-\dfrac{h}{2}\dfrac{\dd^2}{\dd r^2}$ can be paralleled to the  realized volume operator that increase/decrease the return from $0$ to $r$. 

%

In real markets, realized trading volume gradually emerges from latent trading volume, which, in turn, is determined by the potential supply-demand gap (SDG). The potential SDG does not immediately dissipate after a single transaction; instead, it gradually diminishes with future inflows or outflows. Let $Z$ represent the potential SDG,  that is the absolute value of difference between potential demand $Z^{S}$ and potential supply $Z^{D}$. 
%
%

%

After all, the potential SDG is reflected by unexecuted orders placed by passive traders. By categorizing all passive traders into three major classes, a further breakdown of the potential SDG can be achieved.

\begin{itemize}
\item[(1)] Rational speculators, denoted as $a$. This category of investors engages in trend trading where the continuous variation of returns amplifies the potential SDG, while the reversal of returns diminishes the potential SDG.
Mathematically, ${\dd Z_a} = k_a(r)\, r\, \dd r$, where $k_a(r)  > 0$. Additionally, risk aversion for asset hedging weakens the dependence of the supply-demand gap on the magnitude of returns, expressed as $k_a(r) = a - \lambda_a r^2$.
\item[(2)] Irrational speculators, denoted as $b$. Similar to rational speculators, this category engages in trend trading with ${\dd Z_b}= k_b\, r\, \dd r $, where $k_b > 0$. However, unlike rational speculators, this group lacks risk control capabilities for asset volatility, and $k_b$ is a constant generating positive feedback, i.e., $k_b = b$.
\item[(3)] Noise liquidity providers, denoted as $c$, play a crucial role in contributing to market liquidity. As speculators amplify potential SDG, these providers mitigate the SDG by placing counteractive orders in book. Conversely, they adjust accordingly when the speculators reduce the potential SDG. Hence, theoretically ${\dd Z_c} = -\gamma (\dd Z_a + \dd Z_b) = -\gamma(a+b-\lambda_a r^2) r \dd r$, where $\gamma > 0$.  However, in real markets, noise liquidity providers cannot accurately know the potential SDG contributed by speculators; they can only utilize perceived estimates, denoted as $(c-\lambda_c r^2) r\dd r$. Therefore, $\dd Z_c =-\gamma (c-\lambda_c r^2) r\dd r$.
\end{itemize}

%

In conclusion, $\dd Z=\dd Z_a+\dd Z_b+\dd Z_c=[\,(a+b-c\gamma)r+(\gamma \lambda_c-\lambda_a)r^3\,]\,\dd r=(\alpha r+\delta r^3)\dd r$, where $\alpha:=a+b-c\gamma$ and $\delta:=\gamma \lambda_c-\lambda_a$. The increase in potential SDG as the return changes from $0$ to $r$ is given by,
\begin{equation}
		V(r)=\int_{0}^{r} \dd Z=\int_{0}^{r}(\alpha r+\delta r^3)\dd r =\dfrac{\alpha}{2} r^2+\dfrac{\delta}{4} r^4
\end{equation}
Since potential SDG corresponds to unrealized latent volume, we have the following proposition.

%

\begin{thm}
	Assuming the complex-valued probability density of asset returns at $r$ is denoted as $\Psi(r)$, the expected unrealized trading volume that the market anticipates producing in future,  when the return reaches $r$ is given by
	\begin{equation}\label{E:thm2}
	\overline{V}=\int_{-\infty}^{+\infty}\Psi^*(r)\left(\dfrac{\alpha}{2} r^2+\dfrac{\delta}{4} r^4
\right)\Psi(r)\dd r
	\end{equation}
\end{thm}

\begin{proof}
	Considering that $\Psi^*(r)\Psi(r)$ represents  the probability density of achieving a return of $r$, and by leveraging the commutative property within the integral and the definition of mathematical expectation, the proposition is evidently valid.
\end{proof}

%

In the market's pursuit of equilibrium, the process of reducing SDG involves the continuous conversion between realized trading volume $Q$ and unrealized potential trading volume $V$. Dynamically, potential trading volume in the market transforms into realized trading volume, leading to a decrease in $V$ and an increase in $Q$. Conversely, a decrease in current trading volume is equivalent to delaying trades to the future, resulting in a decrease in $Q$ and an increase in $V$. While it is not expected that $V + Q$ will be equal at every level of return, on average, $V + Q$ should remain constant. This is determined by the total available trading money in the market, which remains constant under the condition of balanced money inflow and outflow in the market.

Define $\overline{E}:=\overline{Q}+\overline{V}$ as the intrinsic trading volume. The following proposition holds.

\begin{thm} Given the intrinsic trading volume $\overline{E}$, complex-valued probability density of asset returns, denoted as $\Psi(r)$, satisfies the following equation:
\begin{equation}\label{E:xde}
	-\dfrac{h}{2}\dfrac{\dd^2}{\dd r^2} \Psi(r)+\left(\dfrac{\alpha}{2} r^2+\dfrac{\delta}{4} r^4
\right)\Psi(r)= \overline{E} \,\Psi(r)
\end{equation}
\end{thm}

%

\begin{proof}
	By $\overline{E}:=\overline{Q}+\overline{V}$, we have $\overline{Q}+\overline{V}=\overline{E}\cdot 1=\overline{E}\displaystyle\int_{-\infty}^{+\infty}\Psi^*(r)\Psi(r)\dd r$. Substituting the results from Proposition 1 and Proposition 2 into the left-hand side, we obtain
	\begin{equation*}
		\int_{-\infty}^{+\infty}\Psi^*(r)\left(-\dfrac{h}{2}\dfrac{\dd^2}{\dd r^2}\right)\Psi(r)\dd r+\int_{-\infty}^{+\infty}\Psi^*(r)\left(\dfrac{\alpha}{2} r^2+\dfrac{\delta}{4} r^4
		\right)\Psi(r)\dd r = \int_{-\infty}^{+\infty}\Psi^*(r)\biggl(\overline{E}\, \Psi(r)\biggr)\dd r
	\end{equation*}
	Note that during integration, $\Psi^*(r)$ is constant for $[r, r+\dd r]$, allowing us to remove the integral sign. So, the proposition proved.
\end{proof}


Equation (\ref{E:xde}) from Proposition 3 provides the differential equation that the complex-valued probability $\Psi(r)$ must satisfy. Solving this equation reveals statistical regularities in the distribution of returns, rendering it a quantum probability theoretic asset return model. Interestingly, it is noteworthy that this equation bears a striking resemblance to the nonlinear harmonic Schrödinger equation. However, in deriving this model, there was no assumption that the variation in returns is analogous to the microscopic motion of particles obeying quantum laws, and the theoretical implications of the parameters has no relation to the meaning of parameters in the Schrödinger equation. Nevertheless, the formulated return equation does share certain analogies with the Schrödinger equation: in the Schrödinger equation, the first term on the left corresponds to 'kinetic energy,' the second term to 'potential energy,' and the term $E$ on the right explains the total energy. From the model derivation, it becomes apparent that the unrealized potential trading volume $V$ in the second term is akin to a form of 'potential energy' determining return variations. It reflects supply-demand imbalances, encapsulating information about changes in asset returns. The first term, realized  trading volume $Q$, acts as the 'kinetic energy' converted from the 'potential energy.' It reflects the total contribution of trading volumes from each $\omega$-market. During the mutual conversion of potential and kinetic energy, the proportion of each $\omega$-market undergoes redistribution, leading to changes in the realized trading volume. Given these dynamics, we call the equation (\ref{E:xde}) as the \emph{Schrödinger-like trading equation}.

%

\section{Model Solving and Empirical Analysis}
\subsection{Solution of the asset return model}
To solve this model, we introduce a non-dimensional transformation $\xi =\left(\dfrac{\alpha }{h}\right)^{\frac{1}{4}}r$ and simultaneously redefine the function $\phi(\xi)=\Psi(r)=\Psi\left(\left[\dfrac{h}{\alpha}\right]^{\frac{1}{4}}\xi\right)$. It can be easily shown that if (\ref{E:xde}) holds, then $\phi(\xi)$ satisfies
\begin{equation}\label{E:xde2}
	-\dfrac{\dd^2 }{\dd \xi^2}\phi(\xi)+[\,\xi^2+\lambda \xi^4\,]\phi(\xi)=\Omega\, \phi(\xi)	
\end{equation}

where $\lambda=\dfrac{\delta}{2\alpha}\sqrt{\dfrac{h}{\alpha}}$ and $\Omega=2({\alpha h})^{-\frac{1}{2}}\overline{E}$ . $\Omega$ can be interpreted as the normalized intrinsic trading volume of the asset.

%

1. When $\lambda=0$, that is, $\delta=\gamma\lambda_c-\lambda_a=0$, equation (\ref{E:xde2}) simplifies to the classical linear harmonic oscillator equation found in quantum mechanics. Research indicates that the equation's solutions converge only when $\Omega$ is an odd integer, i.e., $\displaystyle\int_{-\infty}^{+\infty}|\phi({\xi})|^2\dd \xi <+\infty$, adhering to the probabilistic interpretation.

Let $\Omega_n=2n+1, n=0,1,2,\cdots$ and consider the case when $\Omega=\Omega_0=1$. The equation has a particular solution
, $\phi(\xi) \sim e^{-\frac{\xi^2}{2}}$. Regularized, it is expressed as,
\begin{equation*}
	\Psi(r)=\phi(\xi)=\left[\frac{\left(\alpha h^{-1}\right)^{\frac{1}{4}}}{\sqrt{\pi}}\right]^{\frac{1}{2}}e^{-\frac{\xi^2}{2}} =\left[\frac{\left(\alpha h^{-1}\right)^{\frac{1}{4}}}{\sqrt{\pi}}\right]^{\frac{1}{2}}\exp\left({-
\frac{\sqrt{\alpha h^{-1}}\,r^2}{2}}\right)
\end{equation*}
In accordance with the definition of quantum probability, the probability density function of the yield rate distribution is given by
\begin{equation}
	f(r)=|\Psi(r)|^2=\dfrac{1}{\sqrt{2\pi}\sigma}\exp\left(-\frac{r^2}{2\sigma^2}\right),\qquad \text{where}\,\,\sigma=\sqrt[4]{\dfrac{h}{4\alpha}}
\end{equation}
It is evident that when the normalized intrinsic trading volume of the asset takes its minimum value, $\Omega_0=1$, the return rate follows a classical normal distribution. According to the expression for the standard deviation $\sigma$, the model asserts that the volatility of the return positively depends on $h$ while inversely depends on $\alpha$.



For the general case, $\Omega_n=2n+1$, and the general solution to the equation can be expressed as:
\begin{equation}
	\phi(\xi)=A_n e^{-\frac{\xi^2}{2}} H_n(\xi)
\end{equation}
Here, the function $H_n$ represents the Hermite polynomial, satisfying two recursion relations $\dfrac{\dd H_n}{\dd \xi}=2 n H_{n-1}(\xi)$ and $H_{n+1}-2\xi H_n+2n H_{n-1}=0$. The first several Hermite polynomials are as follows:
\begin{align*}
	H_0=1, \qquad H_1=2\xi, \qquad H_2=4\xi^2-2,\qquad H_3=8\xi^3-12\xi,\qquad H_4=16\xi^4-48\xi^2+12
\end{align*}
The coefficient $A_n$ is an undetermined normalization coefficient, ensuring $\displaystyle\int_{-\infty}^{+\infty}|\Psi({\xi})|^2\dd \xi =1$. According to the scaling transformation relationship for the return, $A_n$ can be determined as $A_n=\left[\dfrac{\left({\alpha m}{\,h^{-1}}\right)^{\frac{1}{4}}}{2^n n! \sqrt{\pi}}\right]^{\frac{1}{2}}$.

Figure 4 illustrates the return distribution density functions calculated through the normalized quantum probability modulus for the first six different values of $\Omega$. As observed in the figure, as the intrinsic trading volume in the market increases, the probability distribution gradually shifts from a unimodal distribution corresponding to a normal function to a multimodal distribution. Additionally, for each unit increase in $\Omega$ (an increase of 2), an additional peak is observed in the distribution. Since $\Omega$ can only take discrete values, the intrinsic trading volume $\Omega$ can be considered as the trading \emph{energy level} for asset returns. Higher levels exhibit more pronounced multimodality in the return distribution, and larger peaks are observed at higher return levels. In practice, various factors in market conditions can cause fluctuations in the intrinsic trading volume of assets at different levels, leading to transitions in the return distribution from unimodal to multimodal or vice versa. Thus, over a longer observation period that includes more than one regime shifts, the distribution is composed of mixed distributions corresponding to different intrinsic trading volumes $\Omega$, and the fat-tailed characteristics arise from the distributions generated by larger $\Omega$.


\begin{figure}[H]
\centering
  \includegraphics[width=14cm]{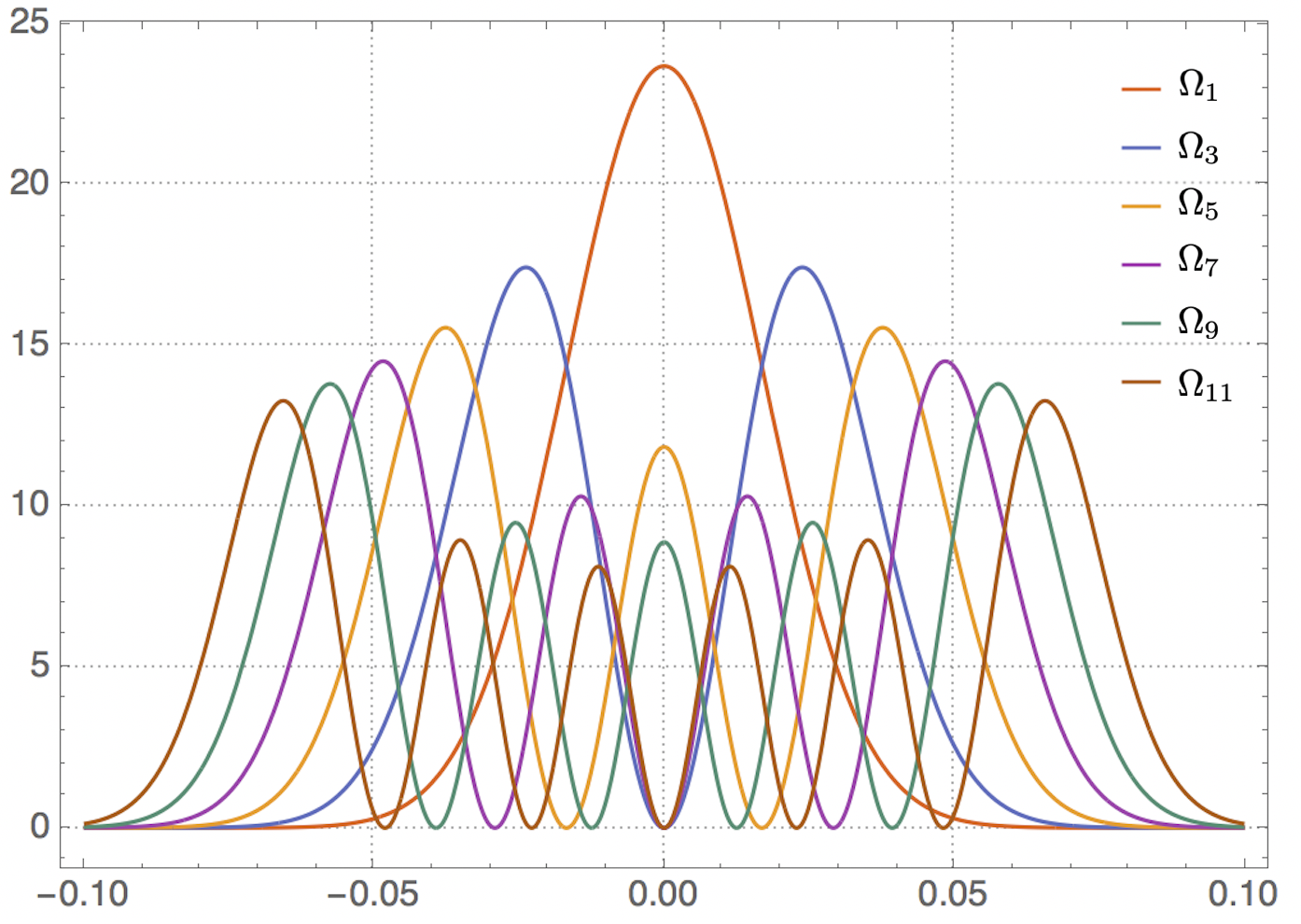}
  \caption{The probability density function of returns, as solved by the theoretical model, with different $\Omega$ corresponding to distinct distribution patterns.}
\end{figure}

2. When $\lambda \neq 0$, equation (\ref{E:xde2}) transforms into the so-called anharmonic equation, which lacks an analytical solution.  However, if the solutions of this equation satisfy the convergence requirement for probability, $\displaystyle\int_{-\infty}^{+\infty}|\phi({\xi})|^2\dd \xi <+\infty$, then $\Omega$ can only take discrete values, satisfying \cite{dasgupta2007}  
\begin{equation}\label{E:3cifang}
	\left(\dfrac{\Omega_n}{2n+1}\right)^3-\left(\dfrac{\Omega_n}{2n+1}\right)=\dfrac{4}{3}\left(1+\dfrac{2n}{3}\right)\lambda,\qquad n=0,1,2,\cdots
\end{equation}
Clearly, $\lambda=0$ yields $\Omega_n=2n+1$. However, when $\lambda \neq 0$, $\Omega$ is not an integer. It is necessary to solve the cubic equation (\ref{E:3cifang}) to find positive real solutions. The difference between two adjacent energy levels is not constant: if $\lambda > 0$, then $\Omega_0>1$ and $\Omega_{n+1}-\Omega_{n}>2$ for all $n\ge0$, and this difference increases with the growth of $n$; conversely, if $\lambda < 0$, then $\Omega_0<1$ and $\Omega_{n+1}-\Omega_{n}<2$ for all $n\ge0$, and this difference gradually decreases with the growth of $n$.

\subsection{Empirical results}
According to the proposed asset return model in this paper, when the standardized intrinsic trading volume transitions from $\Omega_0$ to $\Omega>\Omega_0$ (i.e., from $\overline{E}_0=\sqrt{\dfrac{\alpha h}{4}}\Omega_0$ to $\overline{E}_n=\sqrt{\dfrac{\alpha h}{4}}\Omega_n,\,n>0$), the distribution of returns undergoes a ``phase transition" from a unimodal distribution to a multimodal distribution. Therefore, the empirical work focus of the model is on comparing whether the distribution of returns shifts from unimodal to multimodal at different energy levels between $\Omega_0$ and $\Omega \neq \Omega_0$.


\begin{figure}[t]
\centering
\includegraphics[width=16cm]{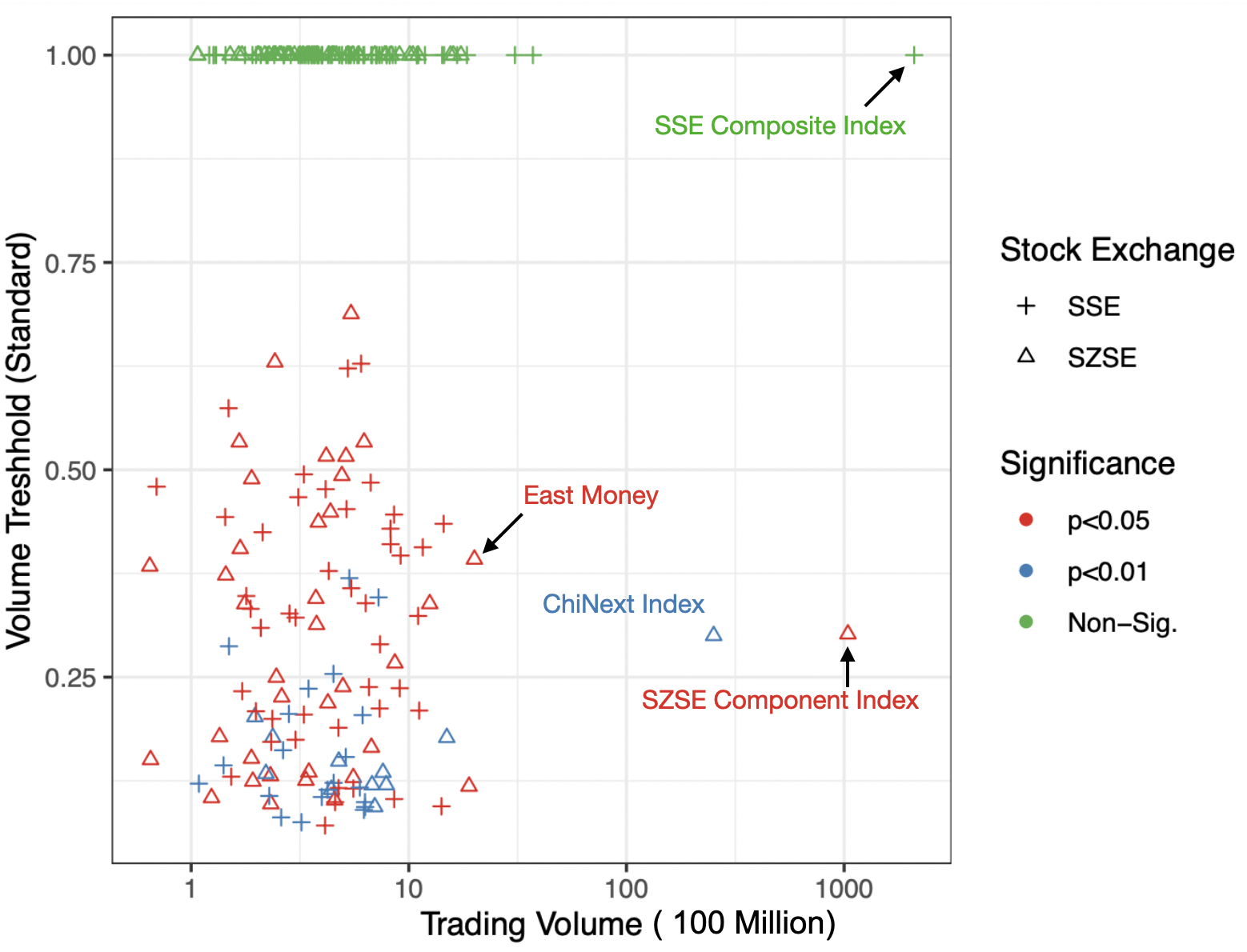}
\caption{Trading 'ground-state energy levels' estimated using the multimodal test method for constituents of the Shanghai and Shenzhen 300 Index and three major stock indices, represented by the normalized trading volume threshold.}
\end{figure}


We selected the constituent stocks of the CSI 300 Index as the focused empirical assets for our model, with data recorded until December 31, 2020. Additionally, we included the three major stock indices commonly used in the A-share market in China—the SSE Composite Index, the SZSE Component Index, and the ChiNext Index—in our calculations, as each stock index can be viewed as an equivalent investment portfolio to its corresponding ETF. The time window spans from January 1, 2011, to December 31, 2020, covering a total of 2432 trading days over the course of a decade. To ensure the robustness of the return distribution, we excluded component stocks with less than four years of trading data, resulting in a total of 234 assets (231 individual stocks and three indices).

Due to the non-stop continuous trading system with opening and closing call auctions in the stock market (which is not active 24 hours a day), price differences arising from non-trading periods can lead to disparities between the previous day's closing price and the opening price on the following day. Taking these factors into account, this study employed daily logarithmic returns, calculated as the difference between the logarithmic closing price and the logarithmic opening price for each day.
For each individual asset, there is one corresponding return data and one trading volume data (corresponding trading value) for each trading day. Although trading volume may not equal the intrinsic trading volume, with a sufficiently large sample size, the mean of trading volume can approximate the intrinsic trading volume.
Given that the model parameters $\sqrt{\dfrac{\alpha h}{4}}$, $\lambda$, and $\overline{E}_0, \overline{E}_1,\cdots$ are all unknown, and considering the multiple unknown parameters, we employed the following algorithm to conduct hypothesis testing for the overall model. For each asset, the specific steps are as follows:

\begin{enumerate}
	\item Let $V^{\text{max}}$ and $V^{\text{min}}$ denote the maximum and minimum daily trading volumes in the empirical window, respectively. Select the 5\% percentile $V^* = V^{\text{min}} + \frac{V^{\text{max}}-V^{\text{min}}}{20}$ as the starting point for the test. Define the subset $\mathscr{G}$ as the set of trading days with volumes less than $V^*$ and $\mathscr{F}$ as the set of trading days with volumes greater than or equal to $V^*$.
	\item Compute the distribution of asset returns within subset $\mathscr{F}$ and conduct a multimodal test using the HY algorithm \cite{cheng1998, HYtest}. Calculate the HY statistic, where $H_0$ represents a unimodal distribution, and $H_1$ suggests the presence of at least two density modes. For statistical significance, conduct 100 Monte Carlo resamplings of returns within each subset $\mathscr{F}$ to derive a p-value. If the p-value is less than $0.05$, identify it as the trading volume threshold for the transition from a unimodal to multimodal distribution, defining $V^*$ as the minimum ``ground-state energy level"$\overline{E}_0$. The remaining energy levels satisfy $V^*<\overline{E}_1<\overline{E}_2<\cdots$.
	\item If the HY test within subset $\mathscr{F}$ cannot be rejected at the $0.05$ significance level, increase the sample splitting point to $V^*\leftarrow V^*+\Delta$, where $\Delta>0$. Based on the updated $V^*$, redivide the subset, obtaining a larger $\mathscr{G}$ and a smaller $\mathscr{F}$. Conduct the hypothesis test in Step 2 with the new subset $\mathscr{F}$.
	\item Repeat Steps 2 and 3 until the subset $\mathscr{F}$, split according to the trading volume threshold, satisfies the HY multimodal test with a p-value<$0.05$. The calculation stops, and $\overline{E}_0=V^*$ is obtained. If, by the time $\mathscr{F}$ contains less than one month (22 days) of trading days, the null hypothesis $H_0$ of a unimodal distribution cannot be rejected, then $\overline{E}_0> V^{\text{max}}$ is determined.
\end{enumerate}

%
%

\begin{figure}[t]
\centering
  \includegraphics[width=16cm]{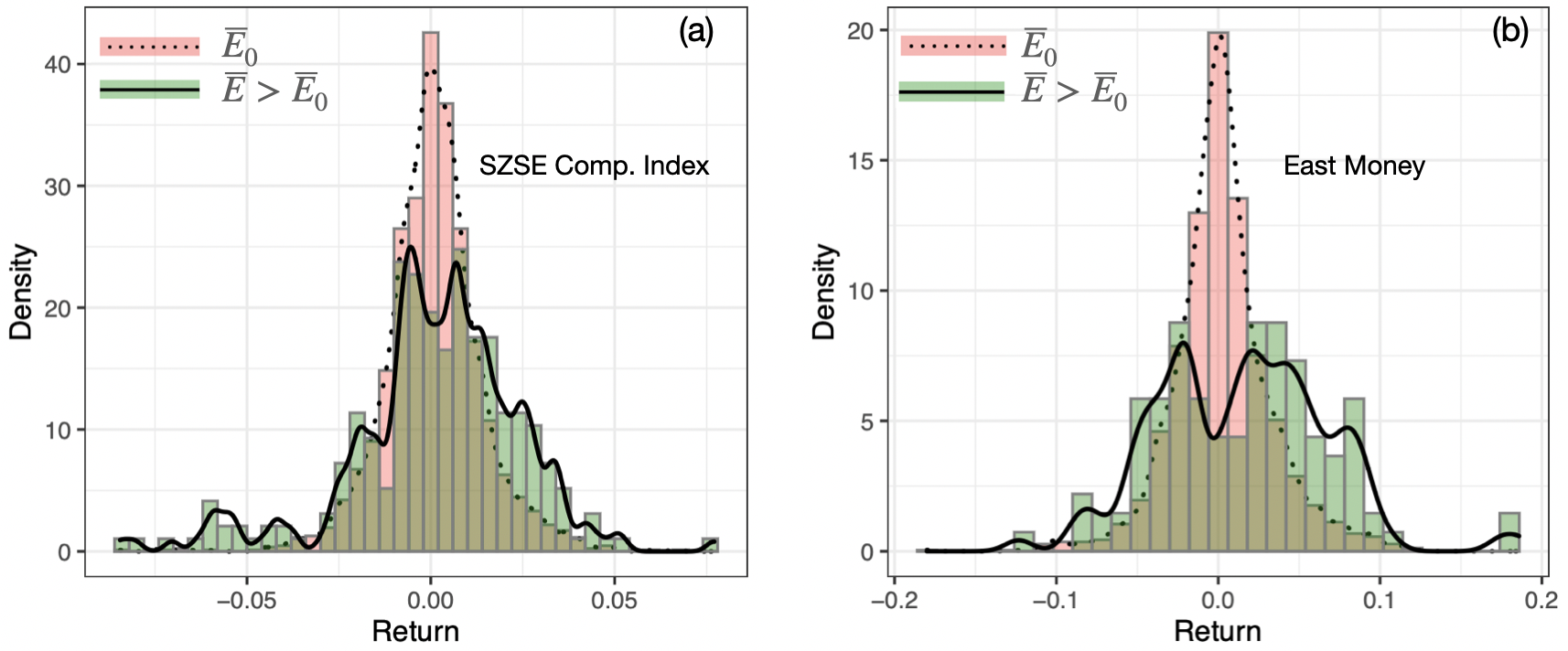}
  \caption{Two selected density functions of financial return distributions are displayed:  (a) SZSE Component Index (b) the stock East Money.}
\end{figure}

Empirical results for 234 assets are presented in Figure 5. The horizontal axis uses the average daily trading volume of each asset during the empirical window as a measure of its trading size (volume). To ensure comparability of the estimated ``ground-state energy level" $\overline{E}_0$ for each asset, the vertical axis is represented by the ratio $\eta=\dfrac{\overline{E}_0}{V^{\text{max}}}$, where $\eta$ is interpreted as the normalized threshold of trading volume for the transition from a unimodal to a multimodal distribution of return probabilities. $\eta\in [0,1]$. A value less than 1 indicates that, within the empirical window, the statistical regularity of the asset's returns is indeed dominated by the trading ``energy level", representing higher trading volumes leading to the multimodal distribution observed in Figure 4. Simultaneously, a smaller ratio suggests a smaller comparable "zero energy level" for the asset. Therefore, at the same turnover rate, the asset's return dynamics are more likely to be in a ``high-energy level" presenting a corresponding multimodal distribution. Empirical results indicate that 114 assets have estimated $\eta$ values less than 1, accounting for approximately 49\% of the total number of empirical assets. This suggests that the existence of trading volume levels cannot be regarded solely as a statistical coincidence, and the quantum probability model of return can provide an explanation for this phenomenon. Additionally, Figure 5 shows that the magnitude of $\eta$ does not exhibit a significant correlation with trading volume, indicating that the ``ground-state energy level" of assets lacks an evident characteristic scale. It does not increase with the intrinsic trading volume of the asset. This result aligns with the model: although the standardized level $\Omega_0$ is around 1 (the magnitude of $\lambda$ may deviate from 1), the size of $\overline{E}_0$ is influenced by the characteristic factor $\sqrt{\dfrac{\alpha h}{4}}$ for each asset. From the empirical results, it can be inferred that this factor depends on the trading size o. Moreover, Figure 4 illustrates that the "ground-state energy level" of each asset is independent of the stock exchange.

To scrutinize empirical findings in detail, Figure 6  further selects the return distributions of two assets at different trading "energy levels," which are considered representative. The chosen assets are the SZSE Component Index (code sz399001) and the stock East Money (code sh300059). The selection criterion is that they have the largest trading volumes in their respective groups and exhibit significant HY test results for  multimodality in the empirical window. Figure 5 has marked the corresponding positions of these two assets. From Figure 6, it is evident that both the SZSE Component Index and East Money exhibit unimodal normal distributions when at the trading "ground-state energy level," precisely corresponding to the distribution shape with $\Omega_0=1$ in the theoretical model solution in Figure 4.

However, when not at the "ground-state energy level", the return distributions show pronounced bimodal or even locally multimodal shapes, with the tails of the distributions originating from non-"ground-state energy level." The distribution shapes correspond to the theoretical model solution in Figure 4 with high energy levels, presenting multimodal distributions $\Omega_1=3,\Omega_2=5,\cdots$. These are mixed distributions of the theoretical model. Due to the ease of reaching the bimodal energy level $\Omega=3$, the distributions predominantly exhibit bimodal shapes, with a small proportion mixed from trimodal or higher-order distributions. Therefore, the phase transition in the return distribution shapes presented in Figure 6 can be well-explained by the quantum probability asset return model.

\section{Conclusion}


Quantum probability is essentially a mathematical structure that extends classical probability from real numbers to complex numbers. It is not a privilege exclusive to microscopic quantum phenomena. Quantifying with complex numbers not only conveys the classical probability meaning of the likelihood of a system state occurring through modulus,  but also conveys additional information through phase. This paper demonstrates that this additional information gives quantum probability an inherent advantage over classical probability in characterizing the multimodal distribution of asset returns.  In this study, we \emph{do not} follow the traditional research path of treating changes in asset returns as the motion of microscopic particles, nor do we presuppose that the movement of financial assets exhibits microscopic quantum effects (such as wave-particle duality, uncertainty principle, etc.). Instead, we directly assign financial interpretations to the mathematical structure of quantum probability based on traders' decisions and financial market trading behavior, associating financial meanings with each component of quantum probability. This helps to derive the equations that quantum probability of asset returns should satisfy. In a nutshell, the complex-valued probability amplitude function $\Psi(r)$ is interpreted as active trading intentions (ATI), where the modulus of $\Psi(r)$ quantifies the intensity of ATI, and the phase of $\Psi(r)$ represents the property of ATI. Using the periodic characteristics and operational properties of the phase, we argue that the phase information of $\Psi(r)$ can not only reflect the repeated switching between long and short in trading decisions but also represent the effects of information interaction among traders. The interaction effects have a mathematically concise expression, referred to as the principle of decomposing the market's total ATI.


Leveraging the interpretation of quantum probability with ATI as well as the decomposing principle,  and also with the introduction of the concept of $\omega$-markets, this paper deduces a differential equation concerning the probability amplitude of returns. It is akin to the Schrödinger equation in quantum mechanics. Although we sidestep the presumption of microscopic quantum effects on returns, the obtained equation still shares many similarities with the Schrödinger equation. Hence, it is referred to as the Schrödinger-like trading equation. The essence of the Schrödinger equation in quantum mechanics is an energy conservation equation: kinetic energy + potential energy = total energy. Similarly, the returns equation in this paper is also a conservation equation: realized trading volume + potential trading volume = total trading volume. By analogy, trading volume is naturally interpreted as the 'energy' of the market during trading. Thus, the concept of intrinsic total trading volume is introduced, with kinetic energy corresponding to all realized trading volume. Due to the symmetry of buying and selling, trading can be seen as a two-sided coin, where one side represents active trading in $\omega$-markets, and the other side represents the execution volume of unbalanced orders from passive providers. Potential energy corresponds to potential trading volume, where one side is the unexecuted unbalanced orders of passive providers, and the other side is the trading volume generated by future active trading from $\omega$-markets.


The form of the Schrödinger-like trading equation implies that the intrinsic trading volume of an asset can only take discrete values, indicating the existence of certain ``energy levels" in asset trading. The theoretical analytical solution reveals that, when the market is at the ground-state energy level," the return follows a normal distribution. As the trading state transitions between different trading energy levels, the distribution of its returns undergoes a phase transition, manifested by changes in the number of density peaks. The distribution density peaks increase in number as the state reaches higher energy levels, while the highest peak shifts towards the tail. Empirical results on the constituents of the CSI 300 Index and stock indices data in the Chinese stock market suggest that the multimodality and fat-tail characteristics of return distributions can be well-explained by the quantum probability asset return model. Additionally, the volatility clustering GARCH effect can also be explained by differences in trading energy levels.  Given these findings, the model holds practical implication for risk management: when trading energy levels magnify, the return distribution deviates from normality and becomes thicker. At such times, utilizing a multi-modal distribution corresponding to the appropriate energy level is a reasonable choice for measuring tail-risk. In practice, the probability distribution of returns is a mixture of distributions corresponding to various trading energy levels. Effectively modeling the intrinsic trading quantity, accurately estimating the energy level states of each asset's trading, and obtaining the correct instantaneous return distribution will be an interesting direction for future research in risk management.

\vspace{2cm}
\noindent{\bf Acknowledgments}
\noindent 
We acknowledge financial support  from  the National Natural Science Foundations of China (Grant No.\,71771086)  

\section*{References}

\bibliography{qmr.bib}

\end{document}